\documentclass{elsarticle}

\usepackage[utf8]{inputenc}
\usepackage[hidelinks]{hyperref}
\usepackage{amsmath} 
\usepackage{amssymb}
\usepackage{amsthm}
\usepackage{graphicx}
\usepackage{subcaption}
\usepackage{color}
\usepackage[]{algorithm2e}

\newcommand{\R}{\mathbb{R}}
\newcommand{\N}{\mathbb{N}}

\newtheorem{assum}{Assumption}
\newtheorem{definition}[assum]{Definition}
\newtheorem{proposition}[assum]{Proposition}
\newtheorem{theorem}[assum]{Theorem}
\newtheorem{remark}[assum]{Remark}

\newtheorem{coro}[assum]{Corollary}
\newtheorem{problem}[assum]{Problem}
\newtheorem{ag}{A/G Obligation}

\graphicspath{{Images/}}

\begin{document}

\begin{frontmatter}



\title{Compositional abstraction refinement\\for control synthesis\tnoteref{title footnote}}
\tnotetext[title footnote]{This work was supported by the H2020 ERC Starting Grant BUCOPHSYS, the EU H2020 AEROWORKS project, the EU H2020 Co4Robots project, the Swedish Foundation for Strategic Research, the Swedish Research Council and the KAW Foundation.}

\author[kth]{Pierre-Jean Meyer}
\author[kth]{Dimos V. Dimarogonas}
\address[kth]{ACCESS Linnaeus Center, School of Electrical Engineering,\\KTH Royal Institute of Technology, SE-100 44, Stockholm, Sweden\\\{pjmeyer, dimos\}@kth.se}

\begin{abstract}
This paper presents a compositional approach to specification-guided abstraction refinement for control synthesis of a nonlinear system associated with a method to over-approximate its reachable sets.
Given an initial coarse partition of the state space, the control specification is given as a sequence of the cells of this partition to visit at each sampling time.
The dynamics are decomposed into subsystems where some states and inputs are not observed, some states are observed but not controlled and where assume-guarantee obligations are used on the uncontrolled states of each subsystem.
A finite abstraction is created for each subsystem through a refinement procedure starting from a coarse partition of the state space, then proceeding backwards on the specification sequence to iteratively split the elements of the partition whose coarseness prevents the satisfaction of the specification.
Each refined abstraction is associated with a controller and it is proved that combining these local controllers can enforce the specification on the original system.
The efficiency of the proposed approach compared to other abstraction methods is illustrated in a numerical example.
\end{abstract}


\end{frontmatter}

\section{Introduction}
\label{sec intro}
In the past decades, a lot of work has been devoted towards model checking and plan synthesis of a finite transition system with respect to high-level specifications such as Linear Temporal Logic~\cite{baier2008principles}.
However, when the system is too large to be handled by such methods in reasonable time or when the system is not finite (e.g.\ continuous dynamics), we must rely on abstraction methods that create a smaller finite system related to the concrete system through a behavioral relationship such as simulation, bisimulation or their alternating and approximate versions~\cite{tabuada2009symbolic}.
Despite the significant progress in both the fields of model checking and abstraction, the link between them is not as straightforward as it appears.
For example, it is possible that the specification is satisfied on the concrete system but not on the chosen abstraction, which would thus require looking for a finer abstraction where this satisfaction can be proved.

This observation is the origin of the development of an interface layer named \emph{abstraction refinement}, whose goal is to use both the dynamics and the specification to automatically obtain an abstract system satisfying the specification by iteratively refining an initial coarse abstraction.
This topic has received many contributions, mainly during the 1990s and early 2000s in the context of model checking for hardware design.
As a consequence, these works primarily focus on verification (as opposed to control synthesis) of a formula in fragments of Computation Tree Logic (CTL) for large but finite systems, where the abstraction procedure is thus only used to reduce the complexity.
The most popular approach is called \emph{CounterExample-Guided Abstraction Refinement} (CEGAR) and consists in exploiting the counterexample provided by the model checker when the abstraction does not satisfy the formula in order to see where the abstraction is too coarse.
These counterexamples can either guide the refinement towards splitting the discrete states of the abstraction where the counterexample originates~\cite{clarke2003counterexample} or improving the partial description of a decomposable system by considering more subsystems~\cite{barner2002symbolic,balarin1993iterative,govindaraju2000counterexample}.
Other refinement methods also consider computing under- or over-approximations of the concrete transition system with iteratively increasing accuracy~\cite{lee1997tearing,pardo1998incremental,lind1999stepwise}, using reachability analysis in a bisimulation algorithm~\cite{chutinan2001verification} or its formula-guided version~\cite{yordanov2013formal}, or uniformly splitting the cells of the state partition based on some error measurement for stochastic systems~\cite{esmaeil2013adaptive}.

In this paper, we present a method for specification-guided abstraction refinement for control synthesis of continuous systems.
In this approach, a coarse abstraction of the system is initially considered and iteratively refined (through repartitioning of the state space) in its elements preventing the satisfaction of the specification.
The problem of abstraction refinement for control synthesis has been mostly unexplored by the methods mentioned in the previous paragraph due to the fact that counterexamples of control problems are much harder to use to guide the refinement since they take the form of trees~\cite{henzinger2003counterexample,fu2014abstractions} (instead of single paths for model checking).
Our approach instead considers a specification-guided approach using reachability analysis to identify the elements of the abstraction that need to be refined in order to find a satisfying controller.
In addition, although some of the previously mentioned works consider infinite state space~\cite{yordanov2013formal,clarke2003abstraction,chutinan2001verification,henzinger2003counterexample}, most results of the above initial literature on abstraction refinement assume that the forward or backward reachable sets can be exactly computed, which is rarely true in systems with a continuous state space.
As a consequence, approaches based on reachability analysis to split \emph{good} and \emph{bad} states into two disjoint sets~\cite{chutinan2001verification,yordanov2013formal,clarke2003counterexample,henzinger2003counterexample,fu2014abstractions} lose a part of their efficiency in such cases.
Our approach thus relies on methods to efficiently compute over-approximations of reachable sets, using for example polytopes~\cite{chutinan1999verification}, oriented hyper-rectangles~\cite{stursberg2003efficient}, ellipsoids~\cite{kurzhanskiy2007ellipsoidal}, zonotopes~\cite{girard2008zonotope}, level sets~\cite{mitchell2000level} or the monotonicity property~\cite{angeli_monotone}.
The use of such over-approximations ensures that a controller synthesized on the abstraction can be applied to the original system in order to satisfy the same specification.
The recent years have seen a renewed interest on the topic of abstraction refinement, but with a focus on control synthesis for continuous systems as in the present paper.
Among the most relevant work, we can see several refinement approaches applied to different types of abstraction.
Indeed, while we use over-approximations of finite-time reachable sets to obtain a non-deterministic abstraction, other approaches consider infinite-time reachability analysis~\cite{nilsson2014incremental}, using some feedback controller on the continuous system to obtain a deterministic abstraction~\cite{mattila2015iterative}, or using sets of finite prefixes to describe abstractions of infinite behaviors~\cite{moor2006learning}.
Another abstraction refinement approach is given in~\cite{tazaki2012discrete}, where the refinement is not guided by the specifications as in our approach but by some behavioral relationship (similar to approximate bisimulation) which is not satisfied on the initial coarse abstraction.
Another relevant method is~\cite{gol2014language}, where the refinement approach is applied on an automaton structure related to the specifications instead of on an abstraction of the system as we do.


As any abstraction-based verification and synthesis problem, this approach is limited to low-dimensional systems due to the classical exponential growth of the abstraction size when the dimension of the state space increases.
This paper thus aims at introducing this abstraction refinement method within a compositional approach where a control objective for the whole system is achieved by working on smaller components~\cite{de1998need}, thus widening the range of applications to systems of larger dimensions or systems only equipped with distributed computation capabilities (e.g.\ multi-robot systems).
More specifically, we adapt the compositional abstraction method presented in~\cite{meyer2015adhs} and~\cite{meyer2015phd} to this abstraction refinement framework.
In this approach, the global system is decomposed into subsystems representing partial descriptions of the dynamics, where some of the states and inputs are not observed and some states are observed to increase the precision of the model but not controlled.
Then, an abstraction can be created for each subsystem using the abstraction refinement approach and the composition of the controllers synthesized on each of these abstractions can be used to control the original system.
To reduce the conservatism of this compositional approach, we consider an \emph{assume-guarantee} reasoning~\cite{henzinger1998you}, similarly used in e.g., the recent results~\cite{nilsson2016compositional} to synthesize controlled invariant sets and~\cite{dallal2015compositional} for a symbolic control synthesis using small-gain theorem.
With such reasoning, the abstraction of each subsystem is obtained under the assumption that other subsystems satisfy their own control objectives and the controller synthesized on one subsystem is then used to guarantee that the assumptions of other subsystems hold.
Defining abstraction refinement within a compositional framework has been mostly unexplored in the literature apart from some results on finite systems~\cite{chaki2003automated,henzinger2003thread,bobaru2008automated,komuravelli2012assume} and, to the best of our knowledge, a single contribution on systems with infinite state space (using hybrid automata)~\cite{bogomolov2014assume}.
As opposed to these papers which are all based on the CEGAR method and thus rely on a model checker providing counterexamples to guide the abstraction refinement, our approach only uses reachability analysis in order to detect unsatisfiability of the specification.

The structure of this paper is as follows.
The problem is formulated in Section~\ref{sec problem}.
Section~\ref{sec abstraction} describes the general method to obtain compositional abstractions.
The abstraction refinement algorithm to be applied to each subsystem is presented in Section~\ref{sec refinement}.
Then, Section~\ref{sec compositional} provides the main result that the local controllers can be composed to control the original system.
Finally, a numerical illustration of this method for the temperature regulation in a $8$-room building is presented in Section~\ref{sec simu}.

\section{Problem formulation}
\label{sec problem}

\subsection{Notations}
\label{sub notations}
In this paper, a decomposition of a system into subsystems is considered.
As a result, both scalar and set variables are used as subscript of other variables, sets or functions:
\begin{itemize}
\item lower case letters and scalars give naming information relating a variable, set or function to the subsystem of corresponding index (e.g.\ $x_i$ and $u_i$ are the state and input of the $i$-th subsystem $S_i$);
\item index sets denoted by capital letters are used to represent the projection of a variable on the dimensions contained in this set.
Alternatively, we also use the projection operator $\pi_I$ to project a set or variable on the dimensions contained in $I$ (e.g.\ for $x\in\R^n$ and $I\subseteq\{1,\dots,n\}$, $x_I=\pi_I(x)$).
\end{itemize}

\subsection{System description}
\label{sub system}
We consider a nonlinear control system subject to disturbances described by 
\begin{equation}
\dot x=f(x,u,w),
\label{eq system}
\end{equation}
with state $x\in X\subseteq\R^n$ and bounded control and disturbance inputs $u\in\mathcal{U}\subseteq\R^p$ and $w\in\mathcal{W}\subseteq\R^q$, respectively.
We denote as $\mathbf{U}$ and $\mathbf{W}$ the sets of piecewise continuous control and disturbance inputs  $\mathbf{u}:\mathbb{R}_0^+\rightarrow\mathcal{U}$ and $\mathbf{w}:\mathbb{R}_0^+\rightarrow\mathcal{W}$, respectively.
$\Phi(t,x^0,\mathbf{u},\mathbf{w})$ denotes the state reached by (\ref{eq system}) at time $t\in \mathbb{R}_0^+$ from initial state $x^0\in X$, under control and disturbance functions $\mathbf{u}\in\mathbf{U}$ and $\mathbf{w}\in\mathbf{W}$, respectively.
If the control input is constantly equal to $u\in\mathcal{U}$ over the interval $[0,t]$, we use the notation $\Phi(t,x^0,u,\mathbf{w})$.
The reachable set of (\ref{eq system}) at time $t\in\R_0^+$, from a set of initial states $X^0\subseteq X$ and for a subset of constant control inputs $\mathcal{U}'\subseteq\mathcal{U}$ is defined as
\begin{equation}
RS(t,X^0,\mathcal{U}')=\left\{\Phi(t,x^0,u,\mathbf{w})~\left|~x^0\in X^0,~u\in\mathcal{U}',~\mathbf{w}\in\mathbf{W}\right.\right\}.
\label{eq reachable set}
\end{equation}
Throughout this paper, we assume that we are able to compute over-approximations $\overline{RS}(t,X^0,\mathcal{U}')$ of the reachable set defined in (\ref{eq reachable set}):
\begin{equation}
RS(t,X^0,\mathcal{U}')\subseteq \overline{RS}(t,X^0,\mathcal{U}').
\label{eq reachable set centralized}
\end{equation}
Several methods exist for over-approximating reachable sets for fairly large classes of linear and nonlinear systems, see e.g.~\cite{chutinan1999verification,stursberg2003efficient,kurzhanskiy2007ellipsoidal,girard2008zonotope,mitchell2000level,angeli_monotone}.

Given a sampling period $\tau\in\R_0^+$, the sampled version of system (\ref{eq system}) with piecewise constant control inputs can then be described as a non-deterministic infinite transition system $S=(X,U,\longrightarrow)$ where
\begin{itemize}
\item $X\subseteq\R^{n}$ is the state space,
\item $U=\mathcal{U}$ is the set of inputs,
\item a transition $x\overset{u}{\longrightarrow}x'$ (equivalently written as $x'\in Post(x,u)$) exists if $x'\in RS(\tau,\{x\},\{u\})$, i.e.\ if there exists a disturbance function $\mathbf{w}\in\mathbf{W}$ such that $x'$ can be reached from $x$ exactly in time $\tau$ by applying the constant control $u$ on $[0,\tau]$.
\end{itemize}

Because of the fact that the control objectives in this paper are expressed in discrete time, the continuous-time system (\ref{eq system}) is immediately sampled.
Nevertheless, we note that the analysis is initialized with a continuous-time system to ensure a more general approach where the behavior of the system between sampling times (e.g.\ when $\mathbf{w}$ is not a constant disturbance function) is taken into consideration in the computation of the reachable set (\ref{eq reachable set}).
If a discrete-time system $x^+=F(x,u,w)$ is given instead of (\ref{eq system}), it can be used to replace $S$ and the reachable set operator (\ref{eq reachable set}) can be redefined as $RS(X^0,\mathcal{U}')=F(X^0,\mathcal{U}',\mathcal{W})$.

The choice of the sampling period $\tau$ for general dynamics as in (\ref{eq system}) is a difficult problem which is not treated in this paper and will be the focus of future work.
Some guidelines to choose its value are provided in~\cite{meyer2017ifac} for the particular case of dynamics taking the form: $\dot x=g(x,w)+u$.

\subsection{Specification}
\label{sub spec}
Let the state space $X\subseteq\R^n$ be a $n$-dimensional interval and $P$ a partition of $X$ into smaller intervals.
This partition $P$ needs to be obtained from a Cartesian product of partitions in each dimension $i\in\{1,\dots,n\}$ of $X$.
In what follows, the elements of $P$ are called \emph{cells} of the state space.
We consider a finite sequence of cells $(\sigma^0,\dots,\sigma^r)\in P^{r+1}$ which is used to define a discrete-time specification of the following form:
\begin{equation}
\psi=\sigma^0\wedge\bigcirc\sigma^1\wedge\bigcirc\bigcirc\sigma^2\wedge\dots\wedge\underset{r}{\underbrace{\bigcirc\dots\bigcirc}}\sigma^r,
\label{eq spec}
\end{equation}
where $\bigcirc$ is the temporal operator ``next''~\cite{baier2008principles} corresponding to the time sampling of period $\tau\in\R_0^+$ as in $S$.
The problem of interest can then be formulated as follows.
\begin{problem}
\label{pb}
Find a controller $C:X\rightarrow U$ such that the sampled system $S$ satisfies the specification $\psi$ in (\ref{eq spec}), 
i.e.\ for any finite sequence of states $(x^0,\dots,x^r)$ with $x^0\in\sigma^0$ and such that $x^k\overset{u^k}{\longrightarrow}x^{k+1}$ and $u^k=C(x^k)$ for all $k\in\{0,\dots,r-1\}$, it holds that $x^k\in\sigma^k$ for all $k\in\{0,\dots,r\}$.
\end{problem}

Although focusing only on specifications of the form of $\psi$ in (\ref{eq spec}) may seem restrictive, we can alternatively consider that we are first provided a more general specification expressed as a Linear Temporal Logic (LTL) formula~\cite{baier2008principles} and that $\psi$ is taken as \emph{any} plan satisfying this LTL formula.
Since this paper is mainly focused on the presentation of the general framework for compositional abstraction refinement, we only consider finite plans as in (\ref{eq spec}) which thus correspond to subclasses of LTL formulas satisfiable in finite time, such as co-safe LTL formulas~\cite{kupferman2001model} or formulas defined over finite traces~\cite{de2013linear}.
Such plans can be extracted from these subclasses of LTL formulas as presented in~\cite{meyer2017ifac}.
For more general LTL formulas, satisfying plans take the form of a \emph{lasso} $\psi=\psi_{pref}.(\psi_{suff})^\omega$ composed of two strings in $P$: a finite prefix path $\psi_{pref}$, followed by a finite suffix path $\psi_{suff}$ repeated infinitely often~\cite{baier2008principles}.
Guidelines on how to consider such infinite-length plans within the compositional abstraction refinement approach are presented in~\cite{meyer2017acc}.

\section{Compositional abstractions}
\label{sec abstraction}
In this paper, Problem~\ref{pb} is addressed with a compositional abstraction refinement approach, where the system is decomposed into subsystems, an abstraction is created for each subsystem using abstraction refinement (Section~\ref{sec refinement}) and the obtained local controllers are then composed to obtain a controller of the original system $S$ (Section~\ref{sec compositional}).
In this section, we present the general method adapted from~\cite{meyer2015adhs} and~\cite{meyer2015phd} to obtain compositional abstractions.

\subsection{System decomposition}
\label{sub decomposition}
Consider that we decompose our dynamics (\ref{eq system}) into $m\in\N$ subsystems.
Let $(I_1^c,\dots,I_m^c)$ be a partition of the state indices $\{1,\dots,n\}$ and $(J_1,\dots,J_m)$ a partition of the control input indices $\{1,\dots,p\}$.
As illustrated in Figure~\ref{fig index partition}, subsystem $i\in\{1,\dots,m\}$ can be described using the following sets of indices:
\begin{itemize}
\item $I_i^c$ represents the state components to be controlled;
\item $I_i\supseteq I_i^c$ are all the state components whose dynamics are modeled in the subsystem;
\item $I_i^o=I_i\backslash I_i^c$ are the state components that are only observed but not controlled;
\item $K_i=\{1,\dots,n\}\backslash I_i$ are the remaining unobserved state components considered as external inputs to subsystem $i$;
\item $J_i$ are the input components used for control;
\item $L_i=\{1,\dots,p\}\backslash J_i$ are the remaining control components considered as external inputs to subsystem $i$.
\end{itemize}
The role of all the index sets above can be summarized as follows: for subsystem $i\in\{1,\dots,m\}$, we model the states $x_{I_i}=(x_{I_i^c},x_{I_i^o})$ where $x_{I_i^c}$ are to be controlled using the inputs $u_{J_i}$ and $x_{I_i^o}$ are only observed to increase the precision of the subsystem while $x_{K_i}$ and $u_{L_i}$ are considered as external disturbances.
It is important to note that the subsystems may share common modeled state components (i.e.\ the sets $I_i$ may overlap), though the sets of controlled state components $I_i^c$ and modeled control input components $J_i$ are necessarily disjoints for two subsystems.
Note also that although the control inputs $u_{J_i}$ have an influence over the states $x_{I_i^o}$, these states are said to be \emph{uncontrolled} due to the fact that their behavior is irrelevant to the design of a controller for this subsystem.
\begin{figure}[tb]
\centering
\includegraphics[width=0.5\textwidth]{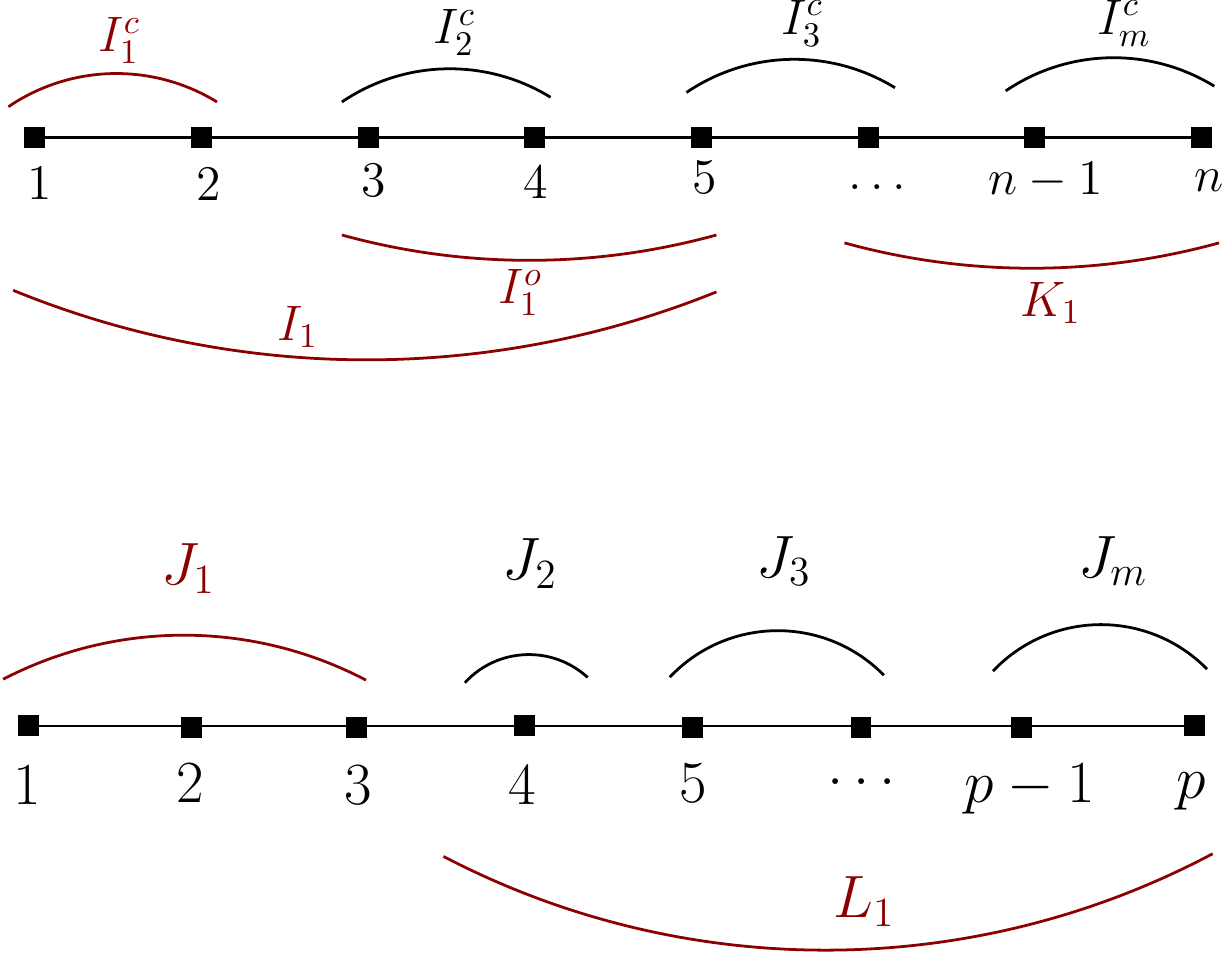}
\caption{Example partition of the state and control dimensions and index sets for subsystem $1$.}
\label{fig index partition}
\end{figure}

%

\subsection{Subsystem's abstraction}
\label{sub subsystem}
For each subsystem $i\in\{1,\dots,m\}$, we want to create a finite abstraction $S_i$ of the sampled system $S$, which models only the state and input components $x_{I_i}$ and $u_{J_i}$, respectively.
$S_i$ will then be used to synthesize a local controller focusing on the satisfaction of the specification for the controlled state components $x_{I_i^c}$ using the modeled control inputs $u_{J_i}$.
The general structure of the abstraction $S_i=(X_i,U_i,\underset{i}{\longrightarrow})$ is as follows.
\begin{itemize}
\item $X_i$ is a partition of $\pi_{I_i}(X)$ into a finite set of intervals called \emph{symbols}. It is initially taken equal to $\pi_{I_i}(P)$ and will then be refined through a procedure detailed in Section~\ref{sec refinement}.
\item $U_i$ is a finite subset of the control set $\pi_{J_i}(\mathcal{U})\subseteq\pi_{J_i}(\R^p)$.
The choice of the discretization of $\pi_{J_i}(\mathcal{U})$ into $U_i$ is free, although it should be noted that having a larger finite set $U_i$ increases the chances to find a satisfying controller while also increasing the computational burden.
\item A transition $s_i\overset{u_i}{\underset{i}{\longrightarrow}}s_i'$ (equivalently written as $s_i'\in Post_i(s_i,u_i)$) exists if $s_i'\cap \pi_{I_i}(RS_i^{AG2}(s_i,u_i))\neq\emptyset$, where the set $RS_i^{AG2}(s_i,u_i)\subseteq X$ defined later in this section is an over-approximation of the reachable set of (\ref{eq system}) based on the partial knowledge available to subsystem $i$.
\end{itemize}

One requirement to be able to compute such an over-approximation of the reachable set is that all variables acting as external disturbances for the considered subsystem need to be bounded.
By definition of system (\ref{eq system}), we know that this is the case for the control and disturbance inputs ($u\in\mathcal{U}$ and $w\in\mathcal{W}$) and in particular the unmodeled control input components satisfy $u_{L_i}\subseteq\pi_{L_i}(\mathcal{U})$.
Moreover, we know that other subsystems will synthesize controllers satisfying the specification for the unobserved and uncontrolled state components ($x_{K_i}$ and $x_{I_i^o}$, respectively) of subsystem $i$: if the state of subsystem $i$ is in the projection $\pi_{I_i}(\sigma^k)$ of some cell $\sigma^k\in P$ involved in the specification $\psi$, then the unobserved states $x_{K_i}$ also start from the projection $\pi_{K_i}(\sigma^k)$ of this cell and the uncontrolled states $x_{I_i^o}$ will reach the next step $\pi_{I_i^o}(\sigma^{k+1})$ of $\psi$.
Assuming that the whole specification $\psi$ in (\ref{eq spec}) is known to each subsystem $i\in\{1,\dots,m\}$, the remark above can be formalized by the following assume-guarantee obligations.
\begin{ag}
\label{ag 1}
For all $x\in X$, $i\in\{1,\dots,m\}$ and $k\in\{0,\dots,r\}$, if $x_{I_i}\in\pi_{I_i}(\sigma^k)$, then $x_{K_i}\in\pi_{K_i}(\sigma^k)$.
\end{ag}
\begin{ag}
\label{ag 2}
For all $i\in\{1,\dots,m\}$, $s_i\in X_i$ and $k\in\{0,\dots,r-1\}$, if $s_i\subseteq\pi_{I_i}(\sigma^k)$, then $\pi_{I_i^o}(RS_i^{AG2}(s_i,u_i))\subseteq\pi_{I_i^o}(\sigma^{k+1})$ for all $u_i\in U_i$.
\end{ag}

The main two differences between these assume-guarantee obligations are that:
\begin{itemize}
\item A/G Obligation~\ref{ag 1} deals with unobserved states $x_{K_i}$, while A/G Obligation~\ref{ag 2} deals with observed but uncontrolled states $x_{I_i^o}$,
\item A/G Obligation~\ref{ag 1} deals with the \emph{initial states} (in $\sigma^k$), while A/G Obligation~\ref{ag 2} deals with the \emph{successors} after one time step (in $\sigma^{k+1}$).
\end{itemize}

\begin{remark}
\label{remark ag}
Unlike traditional assumptions, the assume-guarantee obligations are only taken internally in each subsystem and they do not imply any additional constraints on the overall approach: the control synthesis achieved in each subsystem is exploited to guarantee that the obligations on other subsystems hold.
\end{remark}

We can now finalize the definition of the transition relation of $S_i$, where the over-approximation $RS_i^{AG2}(s_i,u_i)$ is obtained in two steps, each using one of the above assume-guarantee obligations.
Given a symbol $s_i\in X_i$ of $S_i$ with $s_i\subseteq\pi_{I_i}(\sigma^k)$ and a control value $u_i\in U_i$, the first step is to compute an intermediate set $RS_i^{AG1}(s_i,u_i)\subseteq X$ using A/G Obligation~\ref{ag 1} and the operator $\overline{RS}$ in (\ref{eq reachable set centralized}) as follows:
\begin{equation}
RS_i^{AG1}(s_i,u_i)=\overline{RS}(\tau,\sigma^k\cap\pi_{I_i}^{-1}(s_i),\mathcal{U}\cap\pi_{J_i}^{-1}(\{u_i\})).
\label{eq reachable set Si}
\end{equation}
Given a set $s\subseteq X$ such that $s\subseteq\sigma^k$ and a control input $u\in\mathcal{U}$, equations (\ref{eq reachable set}), (\ref{eq reachable set centralized}) and (\ref{eq reachable set Si}) thus give the following inclusions:
\begin{equation}
RS(\tau,s,\{u\})\subseteq \overline{RS}(\tau,s,\{u\})\subseteq RS_i^{AG1}(\pi_{I_i}(s),\pi_{J_i}(u)),
\label{eq reachable set inclusions}
\end{equation}
which implies that the set $RS_i^{AG1}(s_i,u_i)$ is an over-approximation of the reachable set (\ref{eq reachable set}) of $S$ at time $\tau$, from initial states $x\in\sigma^k$ with $\pi_{I_i}(x)\in s_i$ and for control inputs $u\in\mathcal{U}$ with $\pi_{J_i}(u)=u_i$.

The second step towards obtaining $RS_i^{AG2}(s_i,u_i)$ is to use A/G Obligation~\ref{ag 2} to reduce the conservatism of $RS_i^{AG1}(s_i,u_i)$ by exploiting the fact that the control objective of subsystem $i$ only focuses on the controlled state components $x_{I_i^c}$ while the uncontrolled but observed states $x_{I_i^o}$ can be considered to satisfy their specifications due to the control action of other subsystems.
For any $s_i\in X_i$ such that $s_i\subseteq\pi_{I_i}(\sigma^k)$ and $u_i\in U_i$, we define:
\begin{equation}
RS_i^{AG2}(s_i,u_i)=RS_i^{AG1}(s_i,u_i)\cap\{x\in X~|~\pi_{I_i^o}(x)\in\pi_{I_i^o}(\sigma^{k+1})\}.
\label{eq reachable set Si AG2}
\end{equation}
$RS_i^{AG2}$ is thus the same set as the over-approximation $RS_i^{AG1}$ but without the states that violate the specification $\psi$ on the uncontrolled state dimensions $I_i^o$.
This is illustrated in the three cases of Figure~\ref{fig ag2}.
In the first case we have $\pi_{I_i^o}(RS_i^{AG1}(s_i,u_i))\subseteq\pi_{I_i^o}(\sigma^{k+1})$ and A/G Obligation~\ref{ag 2} thus has no effect: $RS_i^{AG2}=RS_i^{AG1}$.
In the second case, the top part of $RS_i^{AG1}$ (in red) is removed after applying A/G Obligation~\ref{ag 2} due to the assumed control action of other subsystems on the state dimensions $I_i^o$.
In the third case, applying A/G Obligation~\ref{ag 2} as in (\ref{eq reachable set Si AG2}) results in $RS_i^{AG2}=\emptyset$, which means that despite the best control actions from other subsystems, the state of the system will always go out of the targeted cell $\sigma^{k+1}$.
This case will need to be treated separately in the next sections by considering as invalid any pair $(s_i,u_i)$ such that $RS_i^{AG2}(s_i,u_i)=\emptyset$.
\begin{figure}[tb]
\centering
\includegraphics[width=0.7\textwidth]{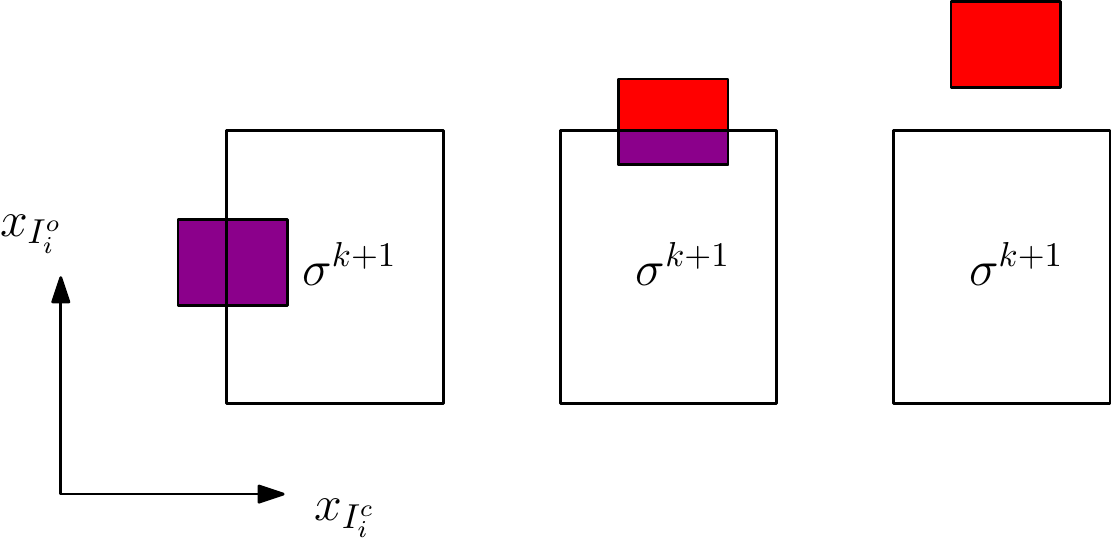}
\caption{Illustration of A/G Obligation~\ref{ag 2} in a 2D example (controlled states $x_{I_i^c}$ on the horizontal axis and uncontrolled but observed states $x_{I_i^o}$ on the vertical axis) with 3 cases of reachable sets colored in red for {\color{red} $RS_i^{AG1}$} and purple for ${\color{red} RS_i^{AG1}}\cap{\color{blue} RS_i^{AG2}}$.}
\label{fig ag2}
\end{figure}

\section{Refinement algorithm}
\label{sec refinement}
In our previous work~\cite{meyer2015phd} where we introduced the compositional abstraction approach summarized in Section~\ref{sec abstraction}, the whole transition system $S_i$ is computed only once for each subsystem as in Section~\ref{sub subsystem} and is mainly based on the knowledge of the dynamics while disregarding most of the influence of the specification $\psi$ (apart from the computation of the sets $RS_i^{AG1}$ and $RS_i^{AG2}$).
As a result, the set of successors $Post_i(s_i,u_i)$ needs to be computed and stored for all symbols $s_i\in X_i$ and all control values $u_i\in U_i$ before being able to start working on the synthesis of a plan satisfying the specification.
In addition, if the chosen partition $P$ of the state space is too coarse, it is likely that no satisfying plan can be found, which then requires the user to pick a new and finer partition and repeat the whole process until the specification is satisfied.

The approach presented in this section aims at addressing both the above problems by:
\begin{itemize}
\item avoiding the computation of the whole abstraction when only a small part is actually relevant to the specification and 
\item automatically adapting the state partition (and hence the set of symbols $X_i$ for each subsystem $i$) if the specification is not initially satisfied.
\end{itemize}

The proposed approach can be described as a compositional abstraction refinement method guided by the control specification $\psi$ in (\ref{eq spec}).
The general idea is that, for each subsystem $i\in\{1,\dots,m\}$, starting from the most coarse abstraction corresponding to the initial partition $X_i=\pi_{I_i}(P)$ as in Section~\ref{sub subsystem}, we identify an element of the abstraction preventing the satisfaction of the specification $\psi$ on the state dimensions $I_i$ of subsystem $i$ and refine it to obtain a more precise abstraction.
This process is then repeated on the new abstraction until $\psi$ can be satisfied by subsystem $i$.

This approach is presented in Algorithm~\ref{algo refinement} and explained below.
For clarity of notations, Algorithm~\ref{algo refinement} is given in the particular case of Assumption~\ref{assum no repeat} below, where for each subsystem, the specification $\psi$ does not visit the same cell twice.
The straightforward modifications required to cover the general case without Assumption~\ref{assum no repeat} are provided at the end of this section.
\begin{assum}
\label{assum no repeat}
For any $k,l\in\{0,\dots,r\}$ such that $k\neq l$ and for all subsystems $i\in\{1,\dots,m\}$ we have $\pi_{I_i}(\sigma^k)\neq\pi_{I_i}(\sigma^l)$.
\end{assum}

\begin{algorithm}[tbh]
  \SetKwFunction{AddToQueue}{AddToQueue}
  \SetKwFunction{FirstInQueue}{FirstInQueue}
  \SetKwFunction{Split}{Split}
  \SetKwFunction{ValidSets}{ValidSets}
  \KwIn{Partition $P$ of $X$,}
  \KwIn{Cell sequence $(\sigma^0,\dots,\sigma^r)\in P^{r+1}$,}
  \KwIn{Discrete control set $U_i$,}
  \KwIn{Partition projection $P_i:P\rightarrow 2^{X_i}$ such that $P_i(\sigma)=\{s_i\in X_i~|~s_i\subseteq\pi_{I_i}(\sigma)\}$.}
  {\bf Initialization: } $X_i=\pi_{I_i}(P)$, $V_i^r=\{\pi_{I_i}(\sigma^r)\}$, $V^r_{iX}=\pi_{I_i}(\sigma^r)$, $Queue=\emptyset$.\\
  \For{$k$ from $r-1$ to $0$}
  {
    $[V_i^k,V^k_{iX},C_i]=\ $\ValidSets$(k,V_{iX}^{k+1})$\\
    $Queue=\ $\AddToQueue($\sigma^k$)\\
    \While{$V_i^k=\emptyset$}
    {
      $\sigma^j=$ \FirstInQueue$(Queue)$\\
      \ForAll{$s_i\in P_i(\sigma^j)\backslash V_i^j$}{$X_i=\ $\Split($s_i$)}
      \For{$l$ from $j$ to $k$}{$[V_i^l,V^l_{iX},C_i]=\ $\ValidSets$(l,V_{iX}^{l+1})$}
    }
  }
  \Return{$X_i$, $\bigcup_{k=0}^{r-1}V_i^k\subseteq X_i$ and $C_i:\bigcup_{k=0}^{r-1}V_i^k\rightarrow {U_i}$}
  \vspace{0.2cm}
\caption{Refinement algorithm for subsystem $i$.\label{algo refinement}}
\end{algorithm}

\paragraph{Inputs}
The refinement method presented in Algorithm~\ref{algo refinement} is described in more details as follows.
We first assume that we are provided with the initial partition $P$ of the state space $X$, the sequence of cells $(\sigma^0,\dots,\sigma^r)\in P^{r+1}$ defining the specification $\psi$ as in Section~\ref{sub spec}, a finite set $U_i$ of control values for subsystem $i$ as in Section~\ref{sub subsystem} and an operator $P_i:P\rightarrow 2^{X_i}$ giving the set of all symbols $s_i\in X_i$ included in the projection $\pi_{I_i}(\sigma)$ of a cell $\sigma\in P$.
For each cell $\sigma^k$ in the sequence $(\sigma^0,\dots,\sigma^r)$ we want to compute the subset $V_i^k$ of symbols in $P_i(\sigma^{k})$ which are considered as valid with respect to the specification $\psi$.
The set $V^k_{iX}$ then corresponds to the projection of $V_i^k$ on the continuous state space $\pi_{I_i}(X)$.

\paragraph{Initialization}
The set of symbols $X_i$ is initially taken as the most coarse partition of the state space $\pi_{I_i}(X)$ (i.e.\ $\pi_{I_i}(P)$) and is then refined during the algorithm when unsatisfaction of $\psi$ is detected.
We proceed backward on the finite sequence $(\sigma^0,\dots,\sigma^r)$, where the target cell $\sigma^r$ is fully valid: $V_i^r=P_i(\sigma^r)=\{\pi_{I_i}(\sigma^r)\}$ and $V_{iX}^r=\pi_{I_i}(\sigma^r)$.
We also initialize a priority queue which will be used to determine which cell of $P$ is to be refined at the next iteration of the algorithm.

\begin{algorithm}[tbh]
  \KwIn{$P$, $(\sigma^0,\dots,\sigma^r)$, $U_i$ and $P_i:P\rightarrow 2^{X_i}$ from Input to Algorithm~\ref{algo refinement},}
  \KwIn{Index of considered cell $k\in\{0,\dots,r-1\}$,}
  \KwIn{Next cell's valid set $V_{iX}^{k+1}$.}
  $V_i^k=\left\{s_i\in P_i(\sigma^{k})~\left|~\exists u_i\in U_i\text{ such that }\emptyset\neq\pi_{I_i}(RS_i^{AG2}(s_i,u_i))\subseteq V_{iX}^{k+1}\right.\right\}$\\
  $V^k_{iX}=\left\{x_i\in \pi_{I_i}(X)~\left|~\exists s_i\in V_i^k\text{ such that }x_i\in s_i\right.\right\}$\\
  $\forall s_i\in V_i^k,~C_i(s_i)$ is chosen in $\left\{u_i\in U_i~\left|~\emptyset\neq\pi_{I_i}(RS_i^{AG2}(s_i,u_i))\subseteq V_{iX}^{k+1}\right.\right\}$\\
  \Return{$V_i^k$, $V^k_{iX}$ and $C_i$}
  \vspace{0.2cm}
\caption{\texttt{ValidSets}. Computes the valid sets and controller for subsystem $i$ at step $k$ of the specification sequence $(\sigma^0,\dots,\sigma^r)\in P^{r+1}$.\label{algo valid}}
\end{algorithm}

\paragraph{External functions}
Algorithm~\ref{algo refinement} calls four external functions.
The function \texttt{ValidSets} looks for the valid symbols and their associated control inputs for a particular step of the specification sequence.
This function is detailed in Algorithm~\ref{algo valid} and explained in the next paragraph.
Functions \texttt{AddToQueue} and \texttt{FirstInQueue} deals with the management of the priority queue and \texttt{Split} represents the refinement of the partition.
Although we provide some guidelines and explanations on the role of these functions in the following paragraphs, they are voluntarily left partially undefined due to their limited influence on the global outcome of the refinement algorithm.
Indeed, if a satisfying refined partition can be found, changing the order in which partition elements are refined or in how many sub-elements they are split only influences how quickly a solution is reached by Algorithm~\ref{algo refinement}.
Some possible heuristics for these functions are detailed for the simulation results in Section~\ref{sub simu results}.
While the main advantage of using a queue to manage the order in which cells are refined is that Algorithm~\ref{algo refinement} always refines the coarsest cells first, other approaches are possible such as associating a cost to the refinement of each cell and refining the cell with minimal cost, as proposed in~\cite{meyer2017ifac}.

\paragraph{Valid sets}
In the main loop of Algorithm~\ref{algo refinement}, assuming we have previously found non-empty valid sets $(V_i^{k+1},\dots,V_i^r)$, we call the function \texttt{ValidSets} for step $k$ of the specification as in Algorithm~\ref{algo valid}.
This function first computes the valid set $V_i^k$ for step $k$ by looking for the symbols in $P_i(\sigma^k)$ for which the over-approximation $RS_i^{AG2}$ of the reachable set is both non-empty and contained in the valid set $V_i^{k+1}$ of the next cell $\sigma^{k+1}$ for at least one value of the discrete control input.
Note that in this particular call when the cell $\sigma^k$ is visited for the first time, $P_i(\sigma^k)$ contains a single element $\pi_{I_i}(\sigma^k)$.
The set $V_{iX}^k$ is taken as the projection of $V_i^k$ on the continuous state space $\pi_{I_i}(X)$.
Then, the controller $C_i$ associates each valid symbol in $V_i^k$ to the \emph{first} of such satisfying control values that has been found.
Algorithm~\ref{algo valid} finally outputs $V_i^k$, $V_{iX}^k$ and $C_i$ to Algorithm~\ref{algo refinement}.
Since the cell $\sigma^k$ is considered here for the first time, we also add it to the priority queue with the function \texttt{AddToQueue}.

\paragraph{Refinement and update}
If the valid set $V_i^k$ is empty, we select (with function \texttt{FirstInQueue}) the first cell $\sigma^j$ of the priority queue and \emph{refine} it.
The refinement is achieved by the function \texttt{Split} and consists in uniformly splitting all the invalid symbols of $P_i(\sigma^j)$ into a number of identical subsymbols (e.g.\ $2$ in each state dimension in $I_i$).
After this, we need to update the valid sets $V_i^j$ and $V_{iX}^j$ and controller $C_i$ of the refined cell $\sigma^j$ using the \texttt{ValidSets} function.
The possibly larger valid set $V_{iX}^j$ obtained after this refinement can then induce a larger valid set at step $j-1$, which in turns influences the following steps.
The refinement and update of the valid set at step $j$ thus requires an update (using function \texttt{ValidSets}) for all other cells from $\sigma^{j-1}$ to $\sigma^k$.
The refined cell $\sigma^j$ can then be moved to any other position in the priority queue (here assumed to be handled by the function \texttt{FirstInQueue}) and these steps are repeated until $V_i^k\neq\emptyset$.
Note that in real implementations, a stopping condition should be added to escape the \emph{while} loop in case the smallest allowed level of partitioning is reached after a certain number of unsuccessful iterations.

\paragraph{Outputs}
The algorithm provides three outputs.
The first one is the refined partition $X_i$ for subsystem $i$.
The second one gathers the sets $V_i^k\subseteq P_i(\sigma^k)\subseteq X_i$ of valid symbols for all $k\in\{0,\dots,r\}$.
Finally, the controller $C_i$ associates a unique control value to each valid symbol.
Note that there is a single value per valid symbol due to the fact that in the presented version of the refinement algorithm, we do not compute the whole set of satisfying control values for a valid symbol, but instead stop looking as soon as one is found and thus avoiding unnecessary computation time.
\begin{remark}
\label{remark Ci}
A potentially interesting direction to be explored in the future is to try to combine the abstraction refinement to some cost minimization problem on the control input.
For this, an alternative version of Algorithm~\ref{algo valid} could be proposed where the whole set $\{u_i\in U_i~|~\emptyset\neq\pi_{I_i}(RS_i^{AG2}(s_i,u_i))\subseteq V_{iX}^{k+1}\}$ of satisfying controls would be computed and thus the output would be a non-deterministic controller $C_i:\bigcup_{k=0}^{r-1}V_i^k\rightarrow 2^{U_i}$.
Once Algorithm~\ref{algo refinement} terminates, we could then select the values which are optimal according to the considered minimization problem.
\end{remark}

As mentioned above, Algorithm~\ref{algo refinement} is presented in the simpler case of Assumption~\ref{assum no repeat}, where the projection of the specification $\psi$ on the state space of each subsystem does not have any duplicated element.
However, the general case without Assumption~\ref{assum no repeat} can easily be covered with the following modifications: 
\begin{itemize}
\item \texttt{AddToQueue} should be adapted so that duplicated cells only appear once in the queue;
\item the controller $C_i$ should not only depend on the current symbol $s_i\in X_i$ but also on the position $k\in\{0,\dots,r\}$ in the specification to know which next cell to target (e.g.\ if $\pi_{I_i}(\sigma^k)=\pi_{I_i}(\sigma^l)$ for some $k\neq l$, $C_i(s_i,k)$ aims towards $\sigma^{k+1}$ while $C_i(s_i,l)$ needs to target $\sigma^{l+1}\neq\sigma^{k+1}$).
\end{itemize}

\section{Composition}
\label{sec compositional}
Algorithm~\ref{algo refinement} in Section~\ref{sec refinement} is applied to each subsystem $i\in\{1,\dots,m\}$ separately.
In this section, we then show that combining the controllers $C_i$ of all subsystems results in a global controller solving Problem~\ref{pb} by ensuring that the sampled system $S$ satisfies the specification $\psi$.

\subsection{Operator for partition composition}
\label{sub compo operator}
Before defining the transition system corresponding to the composition of the abstractions $S_i$ of each subsystem, we need to define an operator to be used in the composition of sets of symbols (either the refined partition $X_i$ or the valid sets $V_i^k$ obtained in Algorithm~\ref{algo refinement}).
There are two main reasons for the introduction of this new operator:
\begin{itemize}
\item the state space of two subsystems $i$ and $j$ may overlap on some dimensions $I_i\cap I_j$, hence a simple Cartesian product is not possible;
\item the refined partition $X_i$ of subsystem $i$ does not necessarily match the partition of other subsystems on their common state dimensions.
\end{itemize}

Intuitively, given two refined sets $X_i$ and $X_j$ as obtained from Section~\ref{sec refinement} with $I_i\cap I_j\neq\emptyset$, we want their composition to be at least as fine as both partitions $X_i$ and $X_j$, which implies that on the common dimensions $I_i\cap I_j$, the composition needs to be at least as fine as the finest of both partitions $\pi_{I_j}(X_i)$ and $\pi_{I_i}(X_j)$.
To ensure the satisfaction of this condition, we can then define the composition operator $\Cap$ as follows:
\begin{equation}
X_i\Cap X_j=\left\{s\in\pi_{I_i\cup I_j}(2^{X})~\left|~
\begin{tabular}{@{}l@{}}
$\exists s_i\in X_i,~\pi_{I_i}(s)\subseteq s_i,$\\ 
$\exists s_j\in X_j,~\pi_{I_j}(s)\subseteq s_j$
\end{tabular}
\right.\right\},
\label{eq compo operator init}
\end{equation}
which provides all the subsets of $\pi_{I_i\cup I_j}(X)$ whose projections onto the state dimensions $I_i$ and $I_j$ are contained in (or equal to) elements of $X_i$ and $X_j$, respectively.
\begin{proposition}
\label{prop composition covering}
If $X_i$ and $X_j$ are partitions of $\pi_{I_i}(X)$ and $\pi_{I_j}(X)$, respectively, then $X_i\Cap X_j$ defined in (\ref{eq compo operator init}) is a covering of $\pi_{I_i\cup I_j}(X)$, i.e.\ $\underset{s\in X_i\Cap X_j}{\bigcup}s=\pi_{I_i\cup I_j}(X)$.
\end{proposition}
\begin{proof}
Let $x\in\pi_{I_i\cup I_j}(X)$.
Since $X_i$ and $X_j$ are partitions, there exists $s_i\in X_i$ and $s_j\in X_j$ such that $\pi_{I_i}(x)\in s_i$ and $\pi_{I_j}(x)\in s_j$, which implies from (\ref{eq compo operator init}) that there exists $s\in X_i\Cap X_j$ such that $x\in s$.
\end{proof}

Although the set in (\ref{eq compo operator init}) contains symbols defined by the finest of both partitions $X_i$ and $X_j$ as mentioned above, it also contains all the subsets included in these symbols, which is not desired if we want the obtained composition $X_i\Cap X_j$ to be a partition of $\pi_{I_i\cup I_j}(X)$.
We thus remove these extra undesired subsets by updating the set $X_i\Cap X_j$ in (\ref{eq compo operator init}) in the following algorithmic expression:
\begin{equation}
\label{eq compo operator final}
X_i\Cap X_j:=X_i\Cap X_j\backslash\{s\in X_i\Cap X_j~|~\exists s'\in X_i\Cap X_j,~s\varsubsetneq s'\},
\end{equation}
where the new composition is based on the old definition of $X_i\Cap X_j$ in (\ref{eq compo operator init}) from which we remove all the elements contained in but not equal to another element of $X_i\Cap X_j$.
In this way, we ensure that we only keep the largest of the elements satisfying (\ref{eq compo operator init}), leading to the desired partition of $\pi_{I_i\cup I_j}(X)$. 
\begin{proposition}
\label{prop composition partition}
If $X_i$ and $X_j$ are partitions of $\pi_{I_i}(X)$ and $\pi_{I_j}(X)$, respectively, then $X_i\Cap X_j$ defined in (\ref{eq compo operator final}) is a partition of $\pi_{I_i\cup I_j}(X)$.
\end{proposition}
\begin{proof}
From Proposition~\ref{prop composition covering} and (\ref{eq compo operator final}), we know that $X_i\Cap X_j$ defined in (\ref{eq compo operator final}) is also a covering since an element $s$ of (\ref{eq compo operator init}) is only removed in (\ref{eq compo operator final}) if it is strictly contained in another element $s'\in X_i\Cap X_j$.

Let $x\in\pi_{I_i\cup I_j}(X)$.
Since $X_i$ and $X_j$ are partitions, there exists $s_i\in X_i$ and $s_j\in X_j$ such that $\pi_{I_i}(x)\in s_i$ and $\pi_{I_j}(x)\in s_j$, and we know that $s_i$ and $s_j$ are unique.
Let now $s,s'\in X_i\Cap X_j$ as defined in (\ref{eq compo operator final}) and such that $x\in s\cap s'$.
From (\ref{eq compo operator init}), we thus have $\pi_{I_i}(s)\subseteq s_i$ and $\pi_{I_j}(s)\subseteq s_j$, which implies that $s\cup s'\in X_i\Cap X_j$.
From (\ref{eq compo operator final}), this implies that $s$ and $s'$ can only be in $X_i\Cap X_j$ if $s=s'=s\cup s'$.
\end{proof}

In the particular case where $I_i$ and $I_j$ do not have any common state dimension ($I_i\cap I_j=\emptyset$), this composition takes the simpler form of a Cartesian product:
$$X_i\Cap X_j=\left\{s\in\pi_{I_i\cup I_j}(2^{X})~|~\pi_{I_i}(s)\in X_i,~\pi_{I_j}(s)\in X_j\right\},$$
where the projections on the dimensions of $I_i$ and $I_j$ need to be exactly an element of $X_i$ and $X_j$, respectively.

\subsection{Composed transition system}
\label{sub compo transition}
We now define the transition system $S_c=(X_c,U_c,\underset{c}{\longrightarrow})$ as the composition of the abstractions of the subsystems for all $i\in\{1,\dots,m\}$ obtained in Algorithm~\ref{algo refinement}.
$S_c$ contains the following elements:
\begin{itemize}
\item $X_c=X_1\Cap\dots\Cap X_m$ is the composition as per (\ref{eq compo operator final}) of the refined partitions for each subsystem.
From Proposition~\ref{prop composition partition}, we know that $X_c$ is a partition of $X$.
\end{itemize}
Due to (\ref{eq compo operator init}), the projection $\pi_{I_i}(s)$ of $s\in X_c$ does not necessarily correspond to a symbol of $X_i$.
However, we know (see proof of Proposition~\ref{prop composition partition}) that there exists a unique symbol $s_i\in X_i$ containing this projection.
Therefore, for each $i\in\{1,\dots,m\}$, we define the decomposition function $d_i:X_c\rightarrow X_i$ such that $d_i(s)=s_i$ is the unique symbol $s_i\in X_i$ satisfying $\pi_{I_i}(s)\subseteq s_i$.
\begin{itemize}
\item $U_c=U_1\times\dots\times U_m$ is the composition of the discretized control sets (which is a simple Cartesian product since they are defined on disjoint dimensions).
\item $\forall s,s'\in X_c,~u\in U_c,~s\overset{u}{\underset{c}{\longrightarrow}}s'\Longleftrightarrow\forall i\in\{1,\dots,m\},~u_{J_i}=C_i(d_i(s))$ and $d_i(s)\overset{u_{J_i}}{\underset{i}{\longrightarrow}}d_i(s')$.
\end{itemize}

Intuitively, the transition $s\overset{u}{\underset{c}{\longrightarrow}}s'$ (equivalently written as $s'\in Post_c(s,u)$) exists when the control input $u\in U_c$ is allowed by the local controllers $C_i$ for all $i\in\{1,\dots,m\}$ and the transition in $S_c$ can be decomposed (using the decomposition functions $d_i:X_c\rightarrow X_i$) into existing transitions for all subsystems.
Consider now the controller $C_c:X_c\rightarrow U_c$ defined by composing the controllers $C_i:X_i\rightarrow U_i$ obtained on the abstraction of each subsystem:
\begin{equation}
\label{eq controller compo}
\forall s\in X_c,~C_c(s)=(C_1(d_1(s)),\dots,C_m(d_m(s))).
\end{equation}
The definition of the transition relation of $S_c$ can then be reformulated as:
$$\forall s,s'\in X_c,~u=C_c(s),~s\overset{u}{\underset{c}{\longrightarrow}}s'\Longleftrightarrow\forall i\in\{1,\dots,m\},~d_i(s)\overset{u_{J_i}}{\underset{i}{\longrightarrow}}d_i(s'),$$
which emphasizes the fact that the composed system $S_c$ defined in this section is restricted to the control inputs allowed by the controller $C_c$.
This can also be expressed by the fact that the set $U_c(s)=\{u\in U_c~|~Post_c(s,u)\neq\emptyset\}$ of allowed controls from a symbol $s$ is included in the singleton $\{C_c(s)\}$.
We prove later (Corollary~\ref{coro non blocking} in Section~\ref{sub compo result}) that we actually have the equality $U_c(s)=\{C_c(s)\}$.

\subsection{Main result}
\label{sub compo result}
To control the original model $S$ with the above controller (\ref{eq controller compo}), the systems $S=(X,U,\longrightarrow)$ and $S_c=(X_c,U_c,\underset{c}{\longrightarrow})$ must satisfy a \emph{feedback refinement relation} defined below and adapted from~\cite{reissig2016}.
In the case where $X_c$ is a partition of $X$ as in this paper, such relation corresponds to a particular case of an \emph{alternating simulation relation}~\cite{tabuada2009symbolic} where $S$ and $S_c$ apply the same control inputs.
\begin{definition}[Feedback refinement]
\label{def simulation}
A map $H:X\rightarrow X_c$ is a feedback refinement relation from $S$ to $S_c$ if it holds: $\forall x\in X,~s=H(x),~\forall u\in U_c(s)\subseteq U,~\forall x'\in Post(x,u),~H(x')\in Post_c(s,u)$.
\end{definition}

This definition means that for any pair $(x,s)$ of matching state and symbol and any control $u$ of the abstraction $S_c$, any behavior of $S$ using this control is matched by a behavior of $S_c$ with the same control.
As a consequence, if a controller is synthesized so that $S_c$ satisfies some specification, then this controller can be \emph{refined} (using the relation $H$ as in Definition~\ref{def simulation}) into a controller ensuring that $S$ satisfies the same specification.
In what follows, we provide such a relation based on the partition $X_c$ of $X$.
\begin{theorem}
\label{th composed input}
The map $H:X\rightarrow X_c$ such that $H(x)=s\Leftrightarrow x\in s$ is a feedback refinement relation from $S$ and $S_c$.
\end{theorem}
\begin{proof}
Let $x\in X$, $s=H(x)\in X_c$ and $u\in U_c(s)\subseteq U$.
By definition of $S_c$, we necessarily have $U_c(s)\subseteq\{C_c(s)\}$ for all $s\in X_c$ since $C_c$ restricts the choice of controls in $S_c$.
If $U_c(s)=\emptyset$, then the condition for feedback refinement in Definition~\ref{def simulation} is trivially satisfied.
Otherwise, we have $u=C_c(s)$ defined as in (\ref{eq controller compo}) which necessarily implies that $x\in\sigma^k$ for some $k\in\{0,\dots,r-1\}$.
Let $x'\in Post(x,u)$, $s'=H(x')$ and denote the decompositions of $s$ and $s'$ as $s_i=d_i(s)$ and $s_i'=d_i(s')$ for all $i\in\{1,\dots,m\}$.
By definition of the over-approximation operator $\overline{RS}$ in (\ref{eq reachable set centralized}), we have $x'\in \overline{RS}(\tau,s,\{u\})$.
With the inclusion in (\ref{eq reachable set inclusions}) and the fact that $\pi_{I_i}(s)\subseteq s_i$, we obtain $x'\in RS_i^{AG1}(s_i,u_{J_i})$ for all $i$.
If $x'\in\sigma^{k+1}$, we immediately have $x_{I_i}'\in s_i'\cap\pi_{I_i}(RS_i^{AG2}(s_i,u_{J_i}))$ and this intersection is thus non-empty, which implies that $s_i'\in Post_i(s_i,u_{J_i})$ for all $i$.
Then, by definition of the transition relation of $S_c$, we have $s'\in Post_c(s,u)$.
On the other hand, if $x'\notin\sigma^{k+1}$, then there exists $l\in\{1,\dots,n\}$ such that $x_l'\notin\pi_l(\sigma^{k+1})$ and there exists a unique subsystem $j\in\{1,\dots,m\}$ such that $l\in I_j^c$.
Therefore we have $x_{I_j^c}'\notin\pi_{I_j^c}(\sigma^{k+1})$ and then $\pi_{I_j}(RS_j^{AG1}(s_j,u_{J_j}))\nsubseteq\pi_{I_j}(\sigma^{k+1})$.
This implies that $u_{J_j}\neq C_j(s_j)$ which contradicts the fact that $u=C_c(s)$.
This case ($x'\notin\sigma^{k+1}$) thus cannot happen and this concludes the proof of the feedback refinement.
\end{proof}

The result in Theorem~\ref{th composed input} thus confirms that restricting the over-approximations with the Assume/Guarantee Obligations~\ref{ag 1} and~\ref{ag 2} is reasonable since it preserves the feedback refinement relation (i.e.\ using the controls of $S_c$, all behaviors of $S$ are matched by behaviors of $S_c$) while allowing us to reduce the conservatism of the approach compared to the general case without assume-guarantee obligations.
The proof of Theorem~\ref{th composed input} can also be used to show that the composed system $S_c$ is non-blocking, i.e.\ if $C_c(s)$ is defined then $Post_c(s,C_c(s))\neq\emptyset$.
\begin{coro}
\label{coro non blocking}
$U_c(s)=\{C_c(s)\}$ for all $s\in X_c$.
\end{coro}
\begin{proof}
We already know that $U_c(s)\subseteq\{C_c(s)\}$ for all $s\in X_c$, so we only need to prove that $C_c(s)\in U_c(s)$ whenever $C_c(s)$ is defined.
This is immediately derived from the proof of Theorem~\ref{th composed input} which states that if $C_c(s)$ exists, then $Post_c(s,C_c(s))\neq\emptyset$.
\end{proof}

These two results can then be exploited to solve Problem~\ref{pb}, where we use the controller $C_c^X:X\rightarrow U$ obtained by combining the controller $C_c:X_c\rightarrow U_c$ of $S_c$ in (\ref{eq controller compo}) and the feedback refinement relation $H:X\rightarrow X_c$ from Theorem~\ref{th composed input}:
\begin{equation}
\label{eq controller S}
\forall x\in X,~C_c^X(x)=C_c(H(x)).
\end{equation}
\begin{theorem}
\label{th controller}
Let $(x^0,\dots,x^r)\in X^{r+1}$ be any finite trajectory of $S$ from an initial state $x^0\in X$ such that $H(x^0)\in V^0_{1}\Cap\dots\Cap V^0_{m}$ and subject to the controller $C_c^X$ in (\ref{eq controller S}), i.e.\ with $x^k\overset{C_c^X(x^k)}{\longrightarrow}x^{k+1}$ for all $k\in\{0,\dots,r-1\}$.
Then we have $x^k\in\sigma^k$ for all $k\in\{0,\dots,r\}$ and $S$ satisfies the specification $\psi$ in (\ref{eq spec}).
\end{theorem}
\begin{proof}
Due to the feedback refinement relation in Theorem~\ref{th composed input}, it is sufficient to prove that the composed system $S_c$ controlled by $C_c$ in (\ref{eq controller compo}) satisfies $\psi$ if it starts in $s^0=H(\mathbf{x}(0))\in V^0_{1}\Cap\dots\Cap V^0_{m}$.
Let $k\in\{0,\dots,r-1\}$ and $s\in X_c$ such that $s\in V^k_{1}\Cap\dots\Cap V^k_{m}$.
The control value $C_c(s)$ in (\ref{eq controller compo}) is thus well defined since we have $d_i(s)\in V^k_i$ for all $i$ and Corollary~\ref{coro non blocking} implies that there exists $s'\in Post_c(s,C_c(s))$.
By definition of $S_c$, this implies that $d_i(s')\in Post_i(d_i(s),C_i(d_i(s)))$ for all $i$.
Then Algorithm~\ref{algo valid} gives that $d_i(s')\in V_i^{k+1}$ for all $i$ and therefore $s'\in V^{k+1}_{1}\Cap\dots\Cap V^{k+1}_{m}$.
\end{proof}

The above result thus states that if Algorithm~\ref{algo refinement} terminates in finite time for all subsystems $i$, the controller $C_c^X$ obtained from composing all subsystem's controllers $C_i$ can be used so that the sampled system $S$ satisfies the specification $\psi$.
On the other hand, if there exists a controller such that $S$ satisfies $\psi$, we cannot always guarantee that Algorithm~\ref{algo refinement} will find partitions $X_i$ for all subsystems $i$ where $\psi$ can be satisfied, due to the fact that the abstractions $S_c$ and $S_i$ are obtained from using \emph{over-approximations} of the reachable set of $S$.

We also wish to emphasize the fact that the approach presented in this paper remains applicable even in the case of unstable dynamics or strongly connected state variables between two subsystems (the state of one subsystem thus creates a large disturbance on the other).
Indeed, such cases would induce very large over-approximations of the reachable sets, thus making the abstraction refinement algorithm unlikely to terminate in reasonable time, but the main result still holds: if Algorithm 1 terminates for all subsystems, the composition of the obtained controllers is a satisfying controller for the original system.

\section{Numerical illustration}
\label{sec simu}
\subsection{System description}
\label{sub simu system}
In this section, we illustrate the proposed compositional refinement procedure on a numerical example representing a temperature regulation problem in an $8$-room building sketched in Figure~\ref{fig ufad}.
The model is inspired by the small-scale experimental building equipped with \emph{UnderFloor Air Distribution} described in~\cite{meyer2014ecc}, where the air in an underfloor plenum is cooled down and sent into each room of the building using controlled fans.
The excess of air in the room is then pushed into a ceiling plenum through exhausts in the fake ceiling and then sent back to the underfloor to be cooled down again.
The walls highlighted in red in Figure~\ref{fig ufad} correspond to open doors and each room is assumed to contain one person whose body heat is considered as a disturbance.
The temperature control of the underfloor is assumed to be realized separately and the model of the building is described by:
\begin{equation}
\label{eq ufad}
\dot T=f(T,u,w),
\end{equation}
where the state $T\in\R^8$ represents the temperature of all rooms, $u\in[-1,0]^8$ is the controlled ventilation in all rooms (using negative values since it has a cooling effect) and $w=[T_u,T_c,T_o,T_b]\in\R^4$ is the disturbance vector containing the underfloor temperature $T_u\in[15,16]\,^{\circ}\mathrm{C}$, the ceiling temperature $T_c\in[26,28]\,^{\circ}\mathrm{C}$, the outside temperature $T_o\in[28,30]\,^{\circ}\mathrm{C}$ and the body heat (assumed to be the same in all rooms) $T_b=37\,^{\circ}\mathrm{C}$.
For simplicity of presentation, we assume that all rooms have the same size and same parameters (wall conduction factor, ventilation factor, \dots)
The temperature variations in room $i\in\{1,\dots,8\}$ can then be described by:
\begin{equation}
\label{eq ufad room}
\frac{dT_i}{dt}=\sum_{j\in\mathcal{N}_i}a_{i,j}(T_j-T_i) + u_ib(T_i-T_u)+c(T_b^4-T_i^4).
\end{equation}
The first term of (\ref{eq ufad room}) models both the conduction through walls and the heat transfers through open doors between room $i$ and a space $j\in\mathcal{N}_i$, where $\mathcal{N}_i$ contains all neighbor rooms of room $i$ as well as the underfloor, ceiling and outside indices $\{u,c,o\}$.
The parameter $a_{i,j}=1*10^{-5}$ is thus taken for a wall between $i$ and $j$, while we consider $a_{i,j}=3*10^{-5}$ for an open door.
The second heat transfer $u_ib(T_i-T_u)$ with $b=2*10^{-4}$ is related to the mass flow rate from the underfloor plenum of temperature $T_u$ to room $i$ of temperature $T_i$, with a ventilation power $u_i\in[-1,0]$.
While each room $i$ is associated to its own ventilation controlled by $u_i$, for the purpose of demonstrating the generality of the proposed approach, we consider that the ventilation in room $6$ is achieved at $75\%$ by $u_6$ and $25\%$ by the control $u_8$ of the neighbor room $8$.
As a result, for $i=6$ in (\ref{eq ufad room}), the second term is replaced by $0.75u_6b(T_6-T_u)+0.25u_8b(T_6-T_u)$.
The last term represents the radiation from the heat source of temperature $T_b=37\,^{\circ}\mathrm{C}$.
Although the chosen parameter $c=10^{-13}$ may appear to be negligible compared to $a_{i,j}$ and $b$, it should be noted that the temperatures in (\ref{eq ufad room}) are to be written in Kelvin degrees, thus resulting in $cT_b^4=9.3*10^{-5}$ which has a similar order of magnitude than other heat transfers.

\begin{figure}[htb]
\centering
\includegraphics[width=0.7\textwidth]{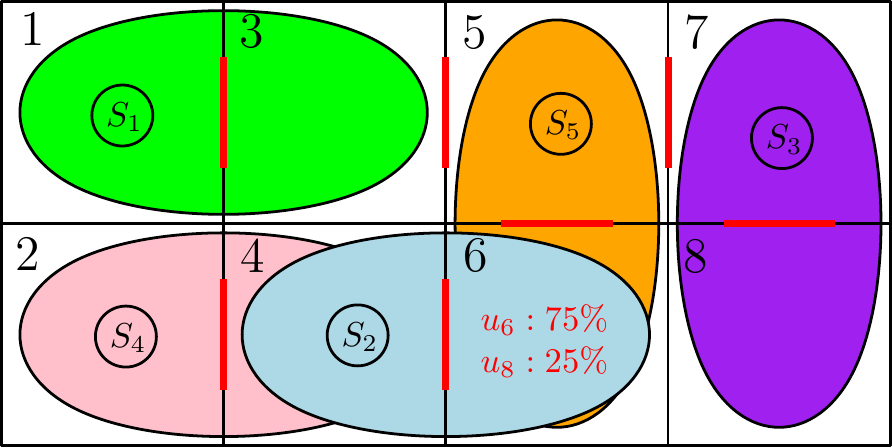}
\caption{Sketch of the $8$-room building decomposed into $5$ subsystems.}
\label{fig ufad}
\end{figure}

The global system (\ref{eq ufad}) is chosen to be decomposed into $5$ subsystems as sketched in Figure~\ref{fig ufad}.
The first $3$ subsystems each cover two rooms $I_1=I_1^c=J_1=\{1,3\}$, $I_2=I_2^c=J_2=\{4,6\}$ and $I_3=I_3^c=J_3=\{7,8\}$ where all the observed states are also controlled ($I_1^o=I_2^o=I_3^o=\emptyset$).
The last $2$ subsystems only aim at controlling a single state each, but also observe an additional state: $I_4=\{2,4\}$, $I_4^c=J_4=\{2\}$, $I_4^o=\{4\}$ and $I_5=\{5,6\}$, $I_5^c=J_5=\{5\}$, $I_5^o=\{6\}$.
Note that the control input $u_8$ (which has a direct influence over the temperature variations in room $6$) will be considered as a disturbance in both subsystems $2$ and $5$ in which the temperature of room $6$ is modeled.

The nonlinear system describing the $8$-room building in (\ref{eq ufad}) and (\ref{eq ufad room}) can be shown to satisfy a monotonicity property as in~\cite{angeli_monotone}, which can be exploited to compute over-approximations of the reachable set, similarly to~\cite{moor2002abstraction} and to~\cite{coogan2015mixed} for the larger class of mixed-monotone systems.

\subsection{Simulation results}
\label{sub simu results}
The considered state space $X=[20,30]^8$ (in Celsius degrees) is partitioned into $5$ elements per dimension, thus resulting in a partition $P$ of $390625$ cells.
The control interval $\mathcal{U}=[-1,0]^8$ is discretized uniformly into $5$ values per dimension: $\{-1,-0.75,-0.5,-0.25,0\}$.
Given an initial state in the cell $[28,30]^8$ and a sampling period $\tau=30$ minutes, our control objective is to reach, within $2$ hours, the conditions describing the following \emph{temperature gradient} (from left to right in the building of Figure~\ref{fig ufad}): $T_1,T_2\in[26,28]$, $T_3,T_4\in[24,26]$, $T_5,T_6\in[22,24]$ and $T_7,T_8\in[20,22]$.
To reach these conditions within $4$ time steps while reducing the energy consumption, we choose the specification $\psi=\sigma^0\sigma^1\sigma^2\sigma^3\sigma^4$ as follows, where the room temperatures are kept in the initial cell as long as possible:
\begin{itemize}
\item $\sigma^0=[28,30]^8$ is the initial cell, 
\item $\sigma^1$ is such that $T_1,T_2,T_3,T_4,T_5,T_6\in[28,30]$ and $T_7,T_8\in[26,28]$,
\item $\sigma^2$ is such that $T_1,T_2,T_3,T_4\in[28,30]$, $T_5,T_6\in[26,28]$ and $T_7,T_8\in[24,26]$,
\item $\sigma^3$ is such that $T_1,T_2\in[28,30]$, $T_3,T_4\in[26,28]$, $T_5,T_6\in[24,26]$ and $T_7,T_8\in[22,24]$,
\item $\sigma^4$ is the final cell described above.
\end{itemize}
Note that the chosen specification $\psi$ does not satisfy Assumption~\ref{assum no repeat}, since we have e.g., $\pi_{I_1}(\sigma^0)=\pi_{I_1}(\sigma^1)=\pi_{I_1}(\sigma^2)=[28,30]^2$.
While Assumption~\ref{assum no repeat} was introduced in Section~\ref{sec refinement} for clarity of notations, the general case without this assumption does not create any problem from an implementation point of view since each element $\pi_{I_i}(\sigma^k)$ is treated independently (e.g.\ with its own partition refinement) even if it represents a previously considered cell.

Algorithm~\ref{algo refinement} is then applied to each subsystem, where the {\tt Split} function uniformly splits a symbol into $2$ subsymbols per dimension and the priority queue is handled as follows: we only refine a cell when no coarser candidate exists and when more than one cell can be refined, we prioritize the one whose last refinement is the oldest.
In Figure~\ref{fig simu}, we display the resulting refined partitions and valid symbols (in red) for each subsystem.
Below, we detail the refinement process in the case of subsystem $3$ in Figure~\ref{fig sub3}.
The bottom left cell $\pi_{I_3}(\sigma^4)$ is fully valid since it is the final step of the specification.
For $\pi_{I_3}(\sigma^3)$, no satisfying control is found to bring the whole cell into $\pi_{I_3}(\sigma^4)$, so it is split into $4$ identical subsymbols, one of which (the bottom left one) is valid.
The next cell $\pi_{I_3}(\sigma^2)$ is in the same situation, but even after splitting it no control can drive any of its subsymbols into the single valid symbol of $\pi_{I_3}(\sigma^3)$.
The next element in $Queue$ is $\pi_{I_3}(\sigma^3)$, so each of its $3$ symbols which are not valid is split into $4$ subsymbols.
Among the obtained $12$ subsymbols of $\pi_{I_3}(\sigma^3)$, $5$ are valid (as in Figure~\ref{fig sub3}) and an update on $\pi_{I_3}(\sigma^2)$ shows that all $4$ symbols of $\pi_{I_3}(\sigma^2)$ can now be controlled towards the new valid set $V_3^3$ of $\pi_{I_3}(\sigma^3)$.
Proceeding with $\pi_{I_3}(\sigma^1)$, we obtain $V_{3X}^1=\pi_{I_3}(\sigma^1)$ after splitting $\pi_{I_3}(\sigma^1)$ once.
The same is then done on $\pi_{I_3}(\sigma^0)$, which also results in $V_{3X}^0=\pi_{I_3}(\sigma^0)$ and thus terminates Algorithm~\ref{algo refinement} since we obtained a non-empty valid set for the initial cell.

\newcommand\proportion{0.49}
\begin{figure}[tbp]
    \centering
    \begin{subfigure}[b]{\proportion\textwidth}
        \includegraphics[width=\textwidth]{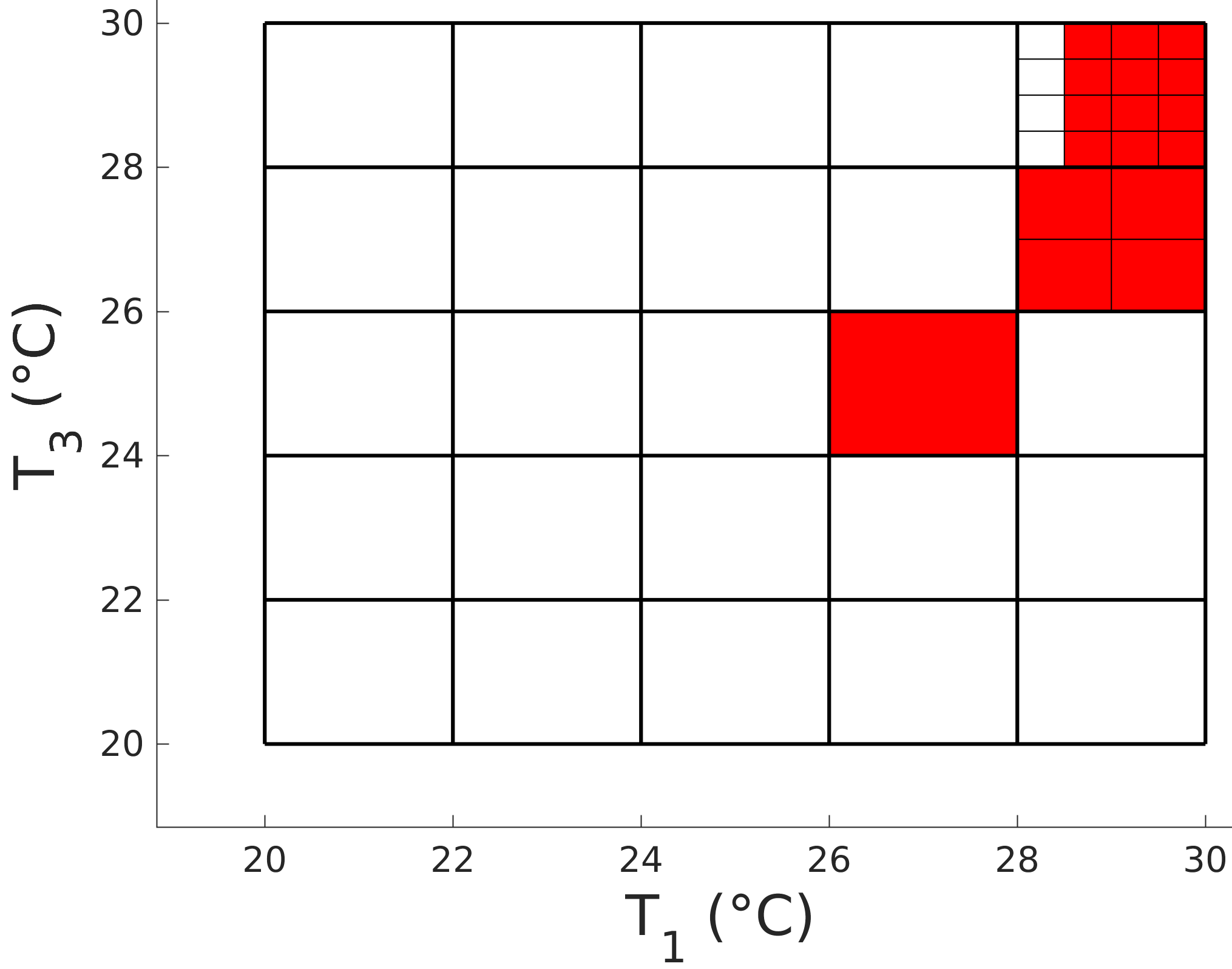}
        \caption{$I_1=I_1^c=J_1=\{1,3\}$}
        \label{fig sub1}
    \end{subfigure}
    \begin{subfigure}[b]{\proportion\textwidth}
        \includegraphics[width=\textwidth]{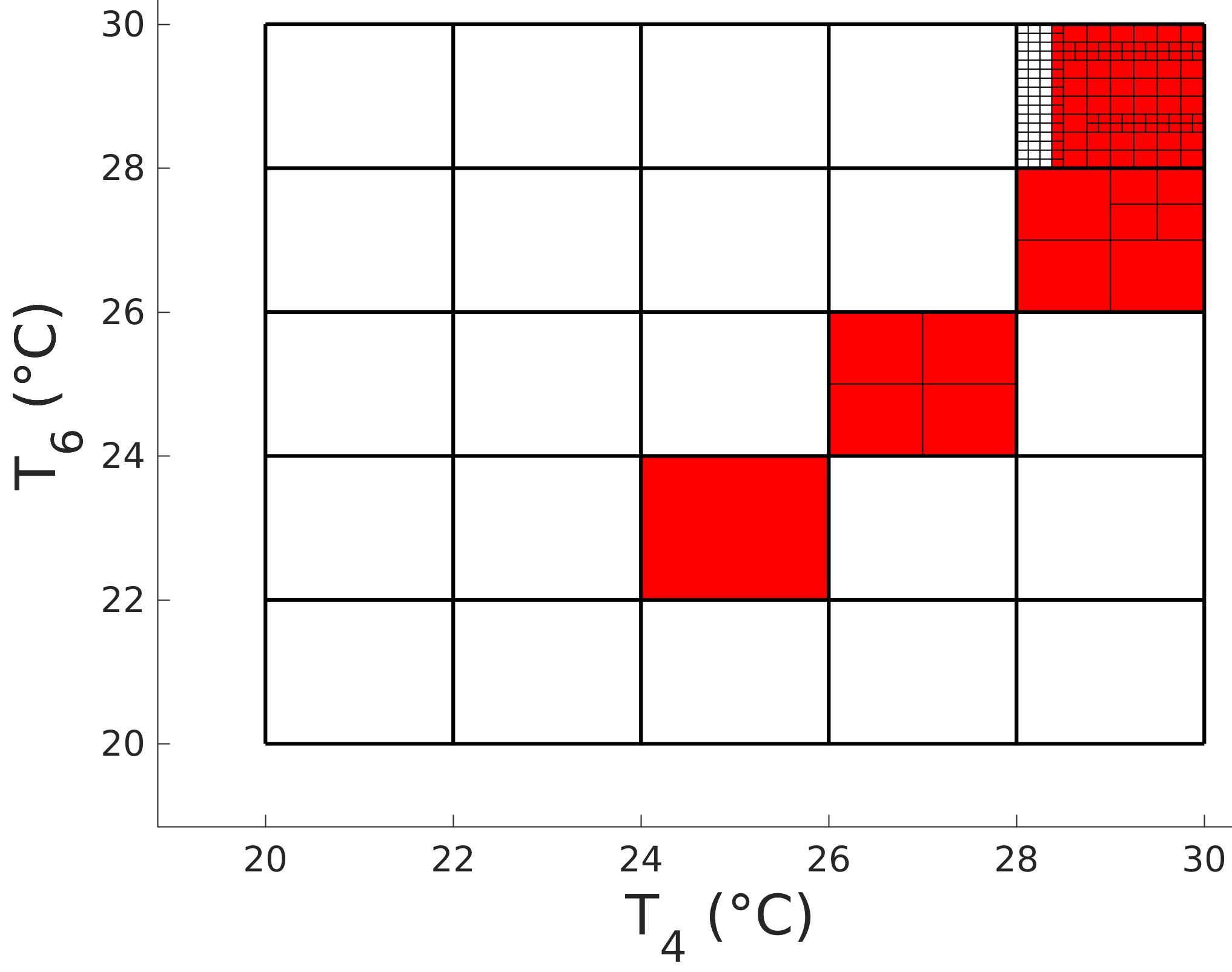}
        \caption{$I_2=I_2^c=J_2=\{4,6\}$}
        \label{fig sub2}
    \end{subfigure}
    
    \begin{subfigure}[b]{\proportion\textwidth}
        \includegraphics[width=\textwidth]{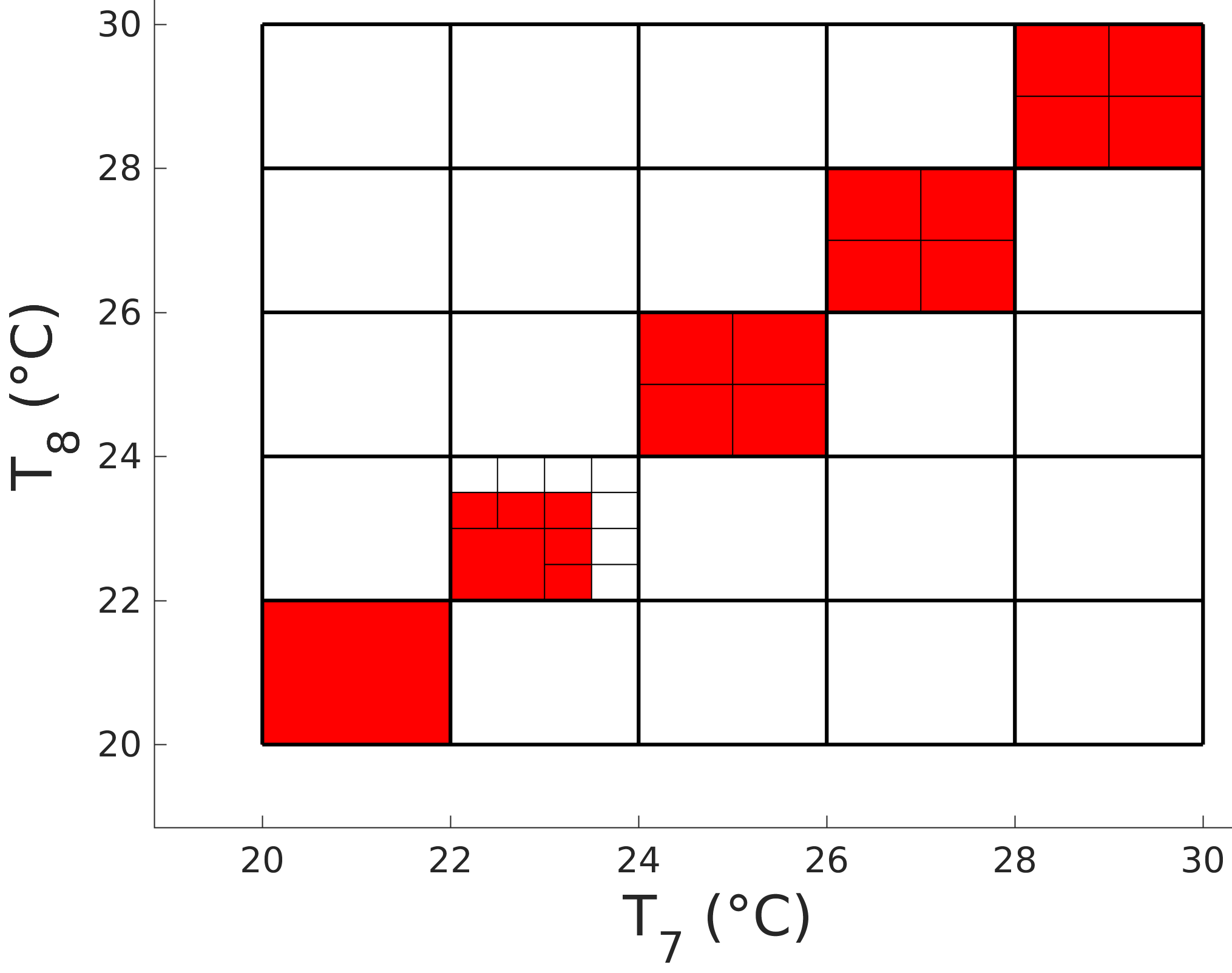}
        \caption{$I_3=I_3^c=J_3=\{7,8\}$}
        \label{fig sub3}
    \end{subfigure}
    \begin{subfigure}[b]{\proportion\textwidth}
        \includegraphics[width=\textwidth]{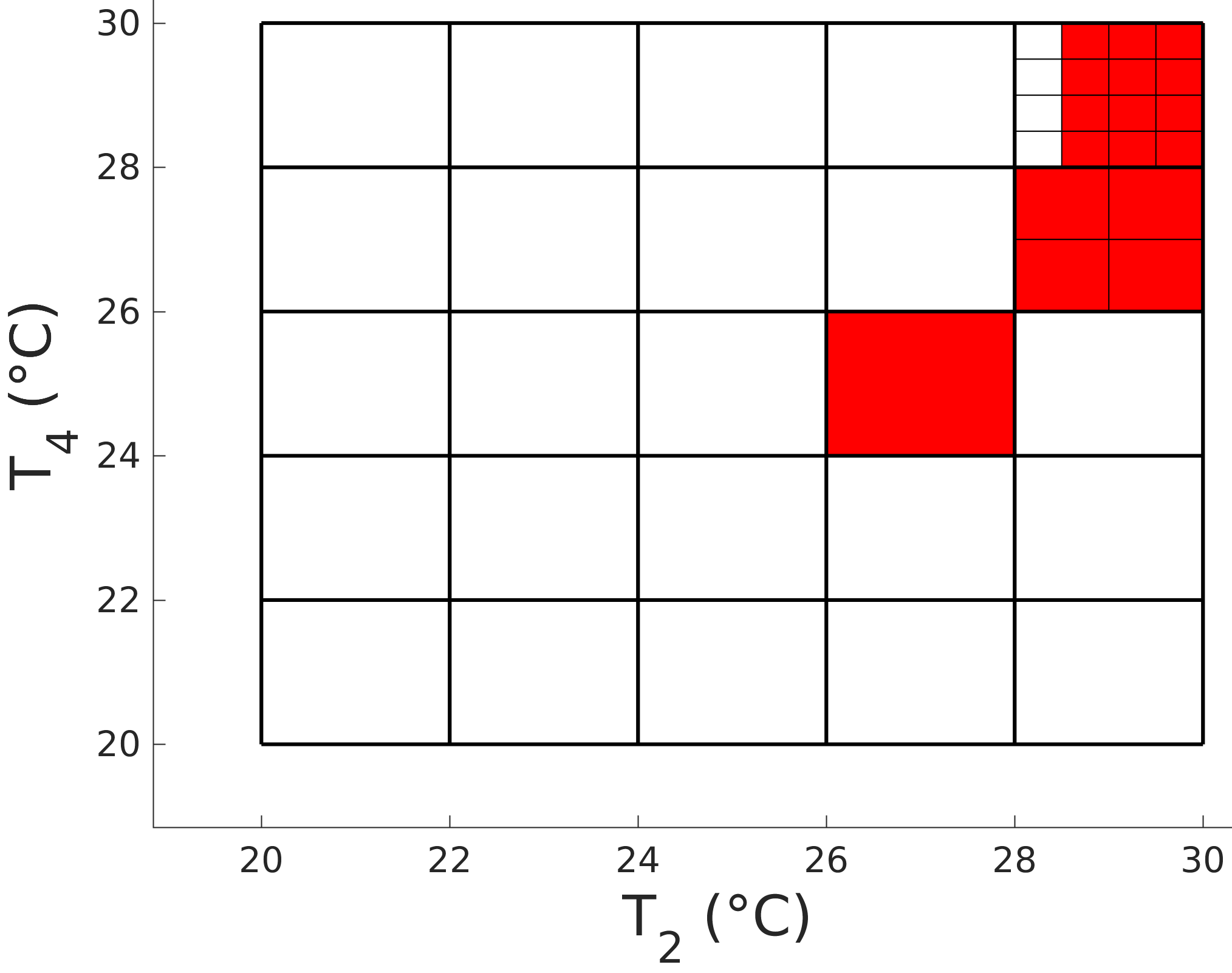}
        \caption{$I_4=\{2,4\}$, $I_4^c=J_4=\{2\}$}
        \label{fig sub4}
    \end{subfigure}
    
    \begin{subfigure}[b]{\proportion\textwidth}
        \includegraphics[width=\textwidth]{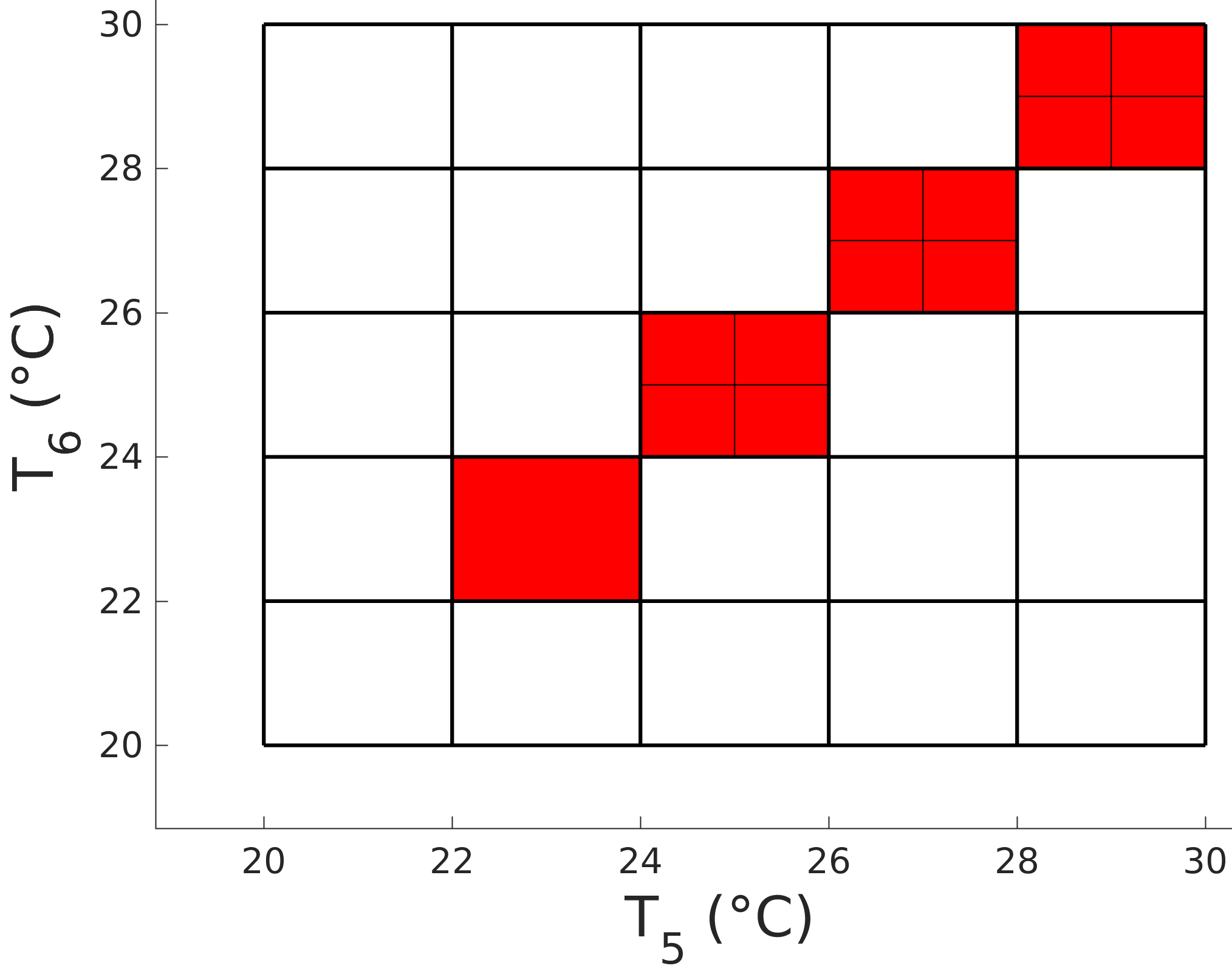}
        \caption{$I_5=\{5,6\}$, $I_5^c=J_5=\{5\}$}
        \label{fig sub5}
    \end{subfigure}
    \caption{Refined partitions and valid symbols (in red) for all $5$ subsystems.}\label{fig simu}
\end{figure}

Algorithm~\ref{algo refinement} is similarly applied to the other $4$ subsystems.
Note that since these subsystems do not satisfy Assumption~\ref{assum no repeat} (unlike subsystem $3$), the results (refined partition and valid symbols) are only partially visible in Figure~\ref{fig simu} due to overlapped cells.
The sequences of refined cells before termination of Algorithm~\ref{algo refinement} are as follows (for clarity of notation, the projections $\pi_{I_i}$ into the relevant state spaces are omitted):
\begin{itemize}
\item $S_1$: $\sigma^3$, $\sigma^2$, $\sigma^1$, $\sigma^2$, $\sigma^0$, $\sigma^1$, $\sigma^0$,
\item $S_2$: $\sigma^3$, $\sigma^2$, $\sigma^1$, $\sigma^2$, $\sigma^1$, $\sigma^0$, $\sigma^0$, $\sigma^1$, $\sigma^2$, $\sigma^0$, $\sigma^1$,
\item $S_4$: $\sigma^3$, $\sigma^2$, $\sigma^1$, $\sigma^2$, $\sigma^0$, $\sigma^1$, $\sigma^0$,
\item $S_5$: $\sigma^3$, $\sigma^2$, $\sigma^1$, $\sigma^0$.
\end{itemize}
We can observe that a larger number of refinements are required in subsystem $2$ due to its weaker control power in room $6$ (only $75\%$ of $u_6$) and the additional disturbance in room $6$ created by $u_8$ (which is not controlled by this subsystem).
Using Matlab on a laptop with a $2.6$ GHz CPU and $8$ GB of RAM, these results after applying Algorithm~\ref{algo refinement} for all subsystems were obtained in $36$ seconds.
As a comparison, for the abstraction refinement algorithm applied in a centralized approach (no decomposition and a single abstraction representing the whole system), the algorithm was still running after more than $64$ hours of computation without yet reaching a result.

\subsection{Complexity comparison}
\label{sub simu complexity}
The complexity reduction induced by the compositional abstraction refinement approach proposed in this paper can be further illustrated by comparing the number of evaluations of the over-approximation operator $RS_i^{AG2}$ in (\ref{eq reachable set Si AG2}) (or $\overline{RS}$ in (\ref{eq reachable set centralized}) for a centralized case with $m=1$) in various abstraction methods.
In Table~\ref{table comparison reachable set}, we compare $4$ such methods: the compositional abstraction refinement from the present paper, the centralized abstraction refinement that can appear as a particular case of this paper with a single subsystem $m=1$ and an abstraction creation without refinement, both in the centralized and compositional cases as in~\cite{meyer2015phd}.
For these comparisons to be meaningful, we consider that the last three methods are computed with a partition corresponding to the finest elements of the refined partition from the first approach (i.e.\ each cell of the initial coarse partition $P$ is split into $2^4=16$ elements per dimension, since $\pi_{I_2}(\sigma^1)$ is refined $4$ times in the simulation above).

\begin{table}[htb]
\centering
\begin{tabular}{ | c || c | c |}
\hline
{\bf \# over-approximations to compute} & Centralized & Compositional\\\hline\hline
Abstraction (no refinement) & $6.55*10^{20}$ & $5.44*10^{5}$\\\hline
Abstraction refinement & $6.74*10^{15}$ & $9933$\\\hline
\end{tabular}
\caption{Number of evaluations of the over-approximation of a reachable set (\ref{eq reachable set centralized}) or (\ref{eq reachable set Si AG2}) for four abstraction methods.}
\label{table comparison reachable set}
\end{table}

For both cases without refinement in the top row of Table~\ref{table comparison reachable set}, the indicated number is the exact number of evaluations of the over-approximation operator required to create the whole abstractions.
For the centralized abstraction refinement, the value is an upper bound on the real number of evaluations since we stop checking other values of the control input as soon as one is found to be valid for a given symbol.
In addition, it is likely that a coarser satisfying partition can be found with the centralized approach due to the considered model being more accurate than the ones used in the compositional approach which only deals with partial representations.
Finally, for the compositional abstraction refinement, the value in Table~\ref{table comparison reachable set} corresponds to the exact number of evaluations in the simulation described in Section~\ref{sub simu results}.

From these results, we can thus observe that the abstraction refinement approach reduces the computational burden in three ways.
Firstly, since the refinement is guided by the specification $\psi$, the obtained abstractions $S_i$ are among the coarsest that provide satisfaction of $\psi$: all cells in $P\backslash\{\sigma^0,\dots,\sigma^r\}$ are left unexplored since they are not relevant to $\psi$ and we stop refining those in $\{\sigma^0,\dots,\sigma^r\}$ as soon as a satisfying path of $S_i$ is found.
The second point which saves both computation time and memory space is the fact that Algorithm~\ref{algo refinement} never actually creates or updates any abstraction: the refinement is done directly on the partitions $X_i$ and the transitions are checked using the over-approximation operator $RS_i^{AG2}$ but the list of successors is never stored.
The last point is that satisfying controllers $C_i$ are obtained from the refinement procedure which means that even at the end of Algorithm~\ref{algo refinement}, we avoid creating the final abstraction and iterating through it for the controller synthesis.

In the particular case of compositional abstraction refinement, an additional complexity reduction is obtained by the decomposition of the dynamics in order to work on lower-dimensional subsystems.
Note however that while the compositional approach is always faster than the centralized one in the case without refinement~\cite{meyer2015phd}, it is not always true in the case of abstraction refinement.
Indeed, the loss of information in the subsystems of the compositional method may require the algorithm to refine more before finding a satisfying solution, which can overcome the complexity reduction from the decomposition.
For this reason, the compositional abstraction refinement approach is particularly adapted to large but weakly coupled systems.





\section{Conclusion}
\label{sec conclu}
In this paper, we have presented a novel approach to abstraction creation and control synthesis in the form of a compositional specification-guided abstraction refinement procedure.
This approach applies to any nonlinear system associated with a method to over-approximate its reachable sets and any control objective formulated as a sequence of locations (in the state space) to visit.
The dynamics are first decomposed into subsystems representing partial descriptions of the system and a finite abstraction is then created for each subsystem through a refinement procedure starting from a coarse partition of the state space.
Each refined abstraction is associated with a controller and we prove that combining these local controllers can enforce the specification on the original system.
The efficiency of the proposed approach compared to other abstraction and synthesis methods was then illustrated in a numerical example showing that this approach is particularly suited to large and weakly coupled systems.

Current efforts aim at combining this approach with specification revision methods into a common framework whose objective is to select for each subsystem which approach is the most advantageous between the abstraction refinement and the specification revision.
Another interesting research direction is to consider a \emph{refinement} on the assume-guarantee obligations (instead of the partition as in this paper), where the controller obtained on a subsystem could be used to shrink the obligations associated with its controlled states, and then send these refined obligations to other subsystems which would then update their abstractions.

\section*{References}

\bibliographystyle{abbrv}
\bibliography{Literature} 

\end{document}